\documentclass[conference]{IEEEtran}
\usepackage[utf8]{inputenc}
\usepackage[T1]{fontenc}
\usepackage{graphicx}
\usepackage[cmex10]{amsmath}
\usepackage{amssymb}
\usepackage{cases}
\usepackage{bm}
\usepackage{amsfonts}
\usepackage{ifthen}
\usepackage{url}
\usepackage{indentfirst}
\usepackage{cite}
\usepackage{caption}
\usepackage{algorithm}
\usepackage{algpseudocode}
\usepackage{booktabs}
\usepackage{subfigure}
\usepackage{color}
\usepackage{multirow}
\usepackage{amsthm}
\usepackage{setspace}

\theoremstyle{remark}
\newtheorem{remark}{Remark}
\theoremstyle{remark}
\newtheorem{theorem}{Theorem}
\theoremstyle{remark}
\newtheorem{corollary}{Corollary}
\theoremstyle{remark}
\newtheorem{lemma}{Lemma}
\theoremstyle{remark}
\newtheorem{definition}{Definition}
\theoremstyle{remark}
\newtheorem{example}{Example}
\theoremstyle{definition}

\newtheorem{transformation}{Transformation}
\makeatletter
\newcommand{\tabcaption}{\def\@captype{table}\caption}
\makeatletter
 \allowdisplaybreaks[4]

\title{Placement Delivery Array Design for Combination Networks with Edge Caching}

\ifCLASSINFOpdf
\else
\fi

\begin{document}
\author{\IEEEauthorblockN{Qifa Yan}
\IEEEauthorblockA{ LTCI, T\'el\'ecom ParisTech\\ 
75013 Paris, France\\
Email: qifa.yan@telecom-paristech.fr}
\and
\IEEEauthorblockN{Mich\`ele Wigger}
\IEEEauthorblockA{LTCI, T\'el\'ecom ParisTech \\
75013 Paris, France\\
Email:  michele.wigger@telecom-paristech.fr}
\and
\IEEEauthorblockN{Sheng Yang\\   }
\IEEEauthorblockA{L2S, CentraleSup\'elec\\
91190 Gif-sur-Yvette, France\\
Email:sheng.yang@centralesupelec.fr}}

\maketitle
\begin{abstract}
  A major practical limitation of the Maddah-Ali-Niesen coded caching techniques is their
  high subpacketization level. For the simple network with a single server and multiple
  users, Yan \emph{et al.} proposed an alternative scheme with the so-called placement
  delivery arrays (PDA). Such a scheme requires slightly higher transmission rates but
  significantly reduces the subpacketization level. In this paper, we extend the PDA
  framework and propose three low-subpacketization schemes for combination networks, i.e.,
  networks with a single server, multiple relays, and multiple cache-aided users that are
  connected to subsets of relays. One of the schemes achieves the cutset lower bound on the
  link rate when the cache memories are sufficiently large. Our other two schemes apply only
  to \emph{resolvable} combination networks. For these networks and for a wide range of cache sizes, the new schemes perform
  closely to the coded caching schemes that directly apply Maddah-Ali-Niesen scheme
  while having significantly reduced subpacketization levels.
\end{abstract}
\section{Introduction}
Caching is a promising  approach to alleviate current network traffics driven by on-demand
video streaming. The idea is to pre-fetch contents during off-peak hours before the actual
user demands, so as to reduce traffic at peak hours when the demands are made.
Therefore, the communication takes place in two phases:  \emph{content placement} at
off-peak hours and \emph{content delivery} at peak hours.

In their seminal work \cite{Maddah2014fundamental}, Maddah-Ali and Niesen modeled the content delivery phase by a shared  error-free link from the single server to all users, and they showed that  delivery traffic in this \emph{shared-link setup} can be highly reduced through a joint design of content placement  and delivery strategy that exploits multicasting opportunities. The scheme is known as \emph{coded caching} and has been extended to various settings, e.g., Gaussian broadcast channels \cite{Gaussian2016},  multi-antenna fading channels \cite{Ngo2017,Elia2017,Caire2017}, or \emph{combination networks} \cite{MJi2015,LiTangISIT2016,Wan2017Bound,combination2017,Zewail2017conf,Wan2018new} as considered in this paper. In a $(h,r)$-combination network,  a single server communicates over dedicated error-free links with $h$ relays and these relays in their turn communicate over dedicated  error-free links with  $h \choose r$ users that have local cache memories. Each user is connected to a different subset of $r$ relays. 
Ji \emph{et al.}~first investigated this network \cite{MJi2015} for the case when $r$ divides $h$ (denoted by $r|h$), and the achievable bound
was improved in \cite{LiTangISIT2016}. In \cite{Wan2017Bound}, Wan \emph{et
al.}~tightened the lower bound under the constraint of uncoded placement, and the achievable
bound for the case when the memory size is small. In  \cite{Zewail2017conf,combination2017,Wan2018asymmetric},   Maxmimum Distance Separable (MDS) codes are applied before placement. In particular, \cite{Zewail2017conf,combination2017} show that the upper bound in \cite{LiTangISIT2016} is achievable for any $(h,r)$ combination network, and \cite{Wan2018asymmetric} shows that  even lower delivery rates are achievable.
  As the results of our work require
memory size larger than that of \cite{Wan2017Bound} and is uncoded placement, we only compare
our results with those from \cite{LiTangISIT2016}.

A key factor that  limits the application of  all forms of coded caching in practice, is the
required high \emph{subpacketization level} \cite{Shanmugam2016}, i.e.,
the number of subpackets must grow exponentially with the number of users.
In contrast, \cite{YanPDA2017,YanBipatite2017,LiTang,Ge2016,linear2017}
proposed new caching schemes that have much lower subpacketization levels but slightly increased transmission rate. A useful tool for representing these new  schemes
is  \emph{placement delivery array (PDA)} introduced in \cite{YanPDA2017}. PDAs
characterize both the (uncoded)~placement and delivery strategies with a single array
\cite{YanPDA2017}, and thus facilitate the design of good caching schemes. 

In this paper,
    we first introduce  \emph{combinational PDAs (C-PDA)} to represent uncoded placement  and delivery strategies for combination networks  in a single array. We also determine the rate, memory, and subpacketization requirements of the caching scheme  corresponding to a given C-PDA.
          Then, for the case when $r|h$, we  describe how any standard PDA with ${h-1 \choose r-1}$ columns can be transformed into a  C-PDA for a $(h,r)$-combination network.
          With this transformation and the previous low-subpacketization schemes for the
          single-shared link setup,  two low-subpacketization schemes for
          $(h,r)$-combination networks are obtained. 
           The performances of the new schemes are close to the scheme in \cite{LiTangISIT2016}, but have significantly lower subpacketization level.
          Finally, for arbitrary $(h,r)$, we propose a C-PDA  for which  the corresponding caching scheme  achieves the cut-set lower bound for sufficiently large cache sizes. 

Due to the space limitation, we only provide sketches of the proofs. For  details, see \cite{C_PDA2018arXiv}.

\textbf{Notations:} We denote the set of positive integers by $\mathbb{N}^+$. For
$n\in\mathbb{N}^+$,  denote the set $\{1,2,\cdots,n\}$ by $[n]$. The  Exclusive OR
operation is denoted by $\oplus$. For a positive real number $x$, $\lceil x\rceil$ is the
least integer that is not less than $x$.
\section{System Model and Preliminaries}\label{sec:model_preliminaries}
%

\subsection{System Model}\label{sec:model}
Consider the $(h,r)$-combination network illustrated in Fig.~\ref{fig:network}, where $h$ and $r$ are positive integers and $r \leq h$. The network comprises of a single server, $h$ relays:
\begin{IEEEeqnarray}{rCl}
\mathcal{H}&=&\{H_1,H_2,\cdots,H_h\},\notag
\end{IEEEeqnarray}
and
$K={h\choose r}$ users  labeled by all the $r$-dimensional subsets of  relay indices $[h]$:
\begin{IEEEeqnarray}{rCl}
	\label{eq:Tset}\textbf{T}&\triangleq&\big\{T\colon\; {T}\subset [h] \;\textnormal{and} \;|{T}|=r\big\}.
	\end{IEEEeqnarray} Each user has a  local cache memory of size $MB$ bits. The relays have no cache memories.
The server can directly access a  library $\mathcal{W}$ of $N$ files,
\begin{IEEEeqnarray}{c}
\mathcal{W}=\{W_1,W_2,\cdots,W_N\},\notag
 \end{IEEEeqnarray}
 where each file $W_n$ consists of $B$ independent and identically uniformly distributed (i.i.d.) random bits. 
 The server can send $R B$ bits to each of the $h$ relays over   an individual  error-free link. Here, $R$ is  the \emph{link rate} (or \emph{rate} for brevity). Each relay can communicate with some of the users. 
 Specifically,  user
$T$ is connected through individual error-free links of rate $R$   to  the $r$ relays with index in $T$, i.e., to relays $\{H_i:i\in{T}\}$.

\begin{figure}[htp!]
  \centering
  \includegraphics[width=0.27\textwidth,height=0.201\textwidth]{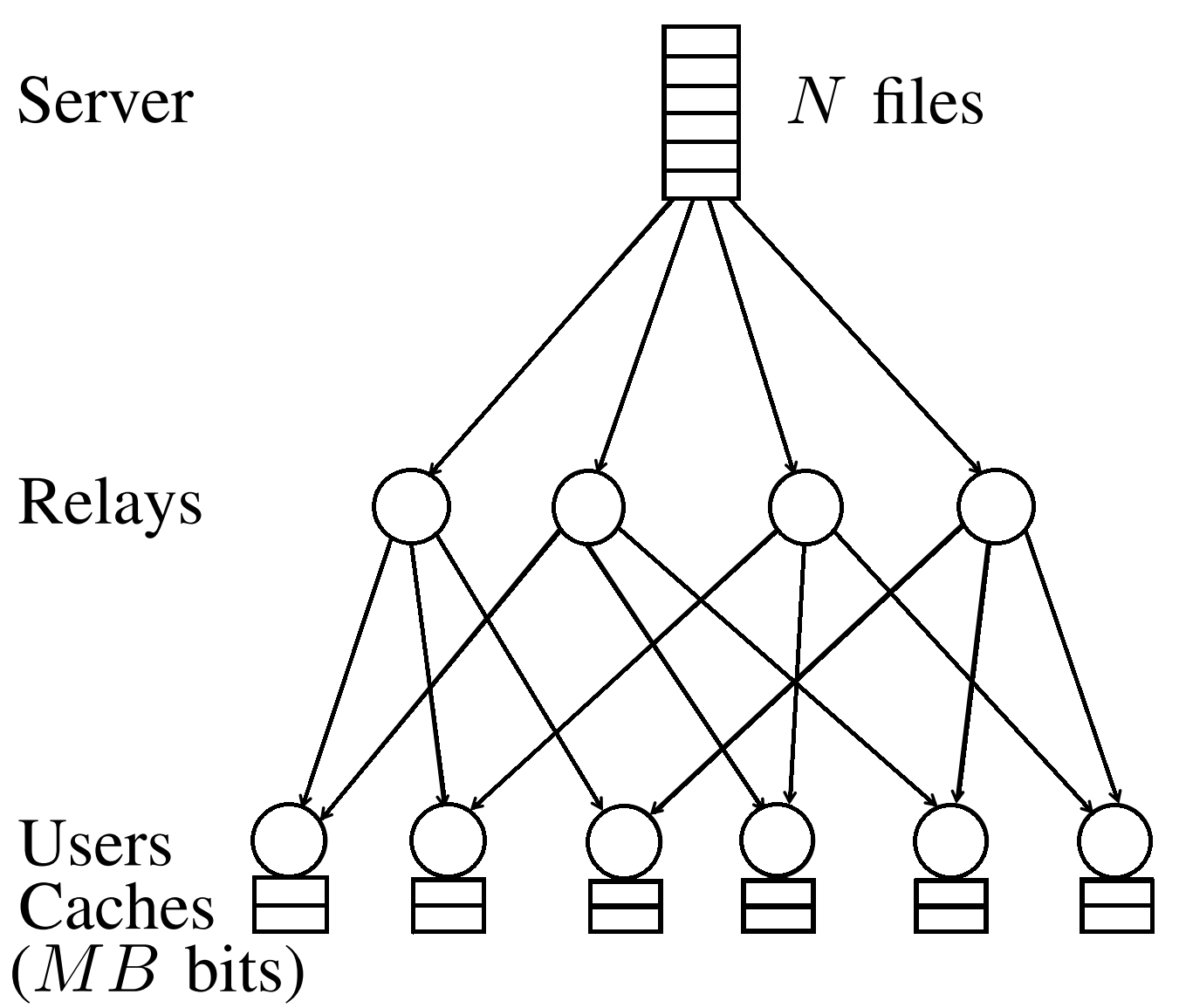}\\
  \caption{A $(4,2)$-combination caching network.}\label{fig:network}
\end{figure}

We now describe the  storage and communication operations. The system operates in two consecutive phases.\\
 \textbf{1. Placement Phase:} In this phase, each user ${T}$  directly accesses to the file library $\mathcal{W}$ and can store  an arbitrary function thereof in its cache memory, subject to the  space limitation of $MB$ bits.   Denote the cached content at user ${T}$ by $Z_{{T}}$, and the set of all cached contents by $\bm Z\triangleq \{Z_{T} \colon T\in\mathbf{T}\}$.\\
\textbf{2. Delivery Phase:} In this phase, each user ${{T}}$ arbitrarily requests a file $W_{d_{T}}$  from the server, where $d_{{T}}\in[N]$. The users' requests
    $\bm{d}\triangleq\{d_{T}\colon  {T}\in\mathbf{T}\}\in[N]^K$ are revealed to all parties, i.e., to   server,  relays, and  users. For each $i\in[h]$, the  server  sends  $R B$ bits to relay $H_i$:
    \begin{IEEEeqnarray}{c}
    X_i=\phi_i( W_1,\ldots, W_N,\bm{Z}, \bm{d}),\notag
    \end{IEEEeqnarray}
    for some   function $\phi_i\colon \mathbb{F}_2^{B \cdot N} \times \mathbb{F}_2^{B\cdot M\cdot K}\times[N]^K\rightarrow\mathbb{F}_2^{B\cdot R}$. Relay $H_i$ forwards  the signal $X_i$ to all  connected users.\footnote{Previous works on combination networks allow the relay to send different arbitrary functions to their connected users. But 
     since the rate of relay-to-users links needs not exceed the rate of the server-to-relays, this apparently more general setup does not allow for better communication strategies.}

At the end of this phase,  each  user $T\in \mathbf{T}$, decodes its requested file $W_{d_{T}}$ based on all its received signals $\bm{X}_T\triangleq \{ X_i\colon i \in T\}$, its cache content $Z_T$, and  demand vector $\bm{d}$:
\begin{IEEEeqnarray}{c}
 \hat{W}_{d_T} = \psi_T (\bm{X}_{T}, Z_T, \bm{d}  ),\notag
\end{IEEEeqnarray}
for some  function $\psi_T\colon \mathbb{F}_2^{B\cdot R\cdot r} \times \mathbb{F}_2^{B\cdot M} \times [N]^K\rightarrow \mathbb{F}_2^B$.

The \emph{optimal worst-case  rate} $R^\star(M)$ is the smallest delivery rate $R$ for which
there exist some placement and delivery strategies so  that the probability of decoding error $\hat{W}_{d_T} \neq W_{d_T}$ vanishes asymptotically as $B\to \infty$ at all the users and for any possible demand $\bm{d}$.

Special focus will be given to $(h,r)$-combination networks with $r|h$.
In this case,  the users can be partitioned into subsets so that in each subset  exactly one  user is connected to a given relay, see     \cite{LiTangISIT2016}.
\begin{definition}[Resolvable Networks] \label{def:resolvable} \emph{A combination network is called \emph{resolvable} if  the user set  $\mathbf{T}$  can be partitioned into subsets $\mathcal{P}_1,\mathcal{P}_2,\cdots,\mathcal{P}_{\tilde{K}}$ so that for all $i\in [\tilde{K}]$ the following two conditions hold:
	\begin{itemize}
		\item If  $T,T'\in \mathcal{P}_i$ and $T\neq T'$, then $T\cap T'=\emptyset$.
		\item  $\bigcup_{T:T\in\mathcal{P}_i}T=[h]$.
	\end{itemize}
Subsets $\mathcal{P}_1,\mathcal{P}_2,\cdots,\mathcal{P}_{\tilde{K}}$ satisfying these conditions are called \emph{parallel classes}.}
\end{definition}


\subsection{Preliminaries: Shared-link Setup and PDAs}
For the purpose of this subsection, consider the original coded caching setup \cite{Maddah2014fundamental}  with a single server and $K$ users each having a cache memory of $M B$ bits. The server is connected to the users through a  shared error-free link of rate $R$. 

Yan \emph{et al.}~\cite{YanPDA2017} proposed to unify the description of uncoded  placement and
delivery strategies for this {shared-link setup} in a single array, called the \emph{placement delivery array (PDA)}.
\begin{definition}[PDA,\cite{YanPDA2017}]\label{def1}\emph{
For  positive integers $K,F, Z$ and $S$, an $F\times K$ array  $\bm{A}=[a_{j,k}]$, $j\in [F], k\in[K]$, composed of a specific symbol $``*"$  and $S$ ordinary symbols
$1,\cdots, S$, is called a \emph{$(K,F,Z,S)$ placement delivery array (PDA)}, if it satisfies the following conditions:
\begin{enumerate}
  \item [C$1$.] The symbol $``*"$ appears $Z$ times in each column;
  \item [C$2$.] Each ordinary symbol occurs at least once in the array;
  \item [C$3$.] For any two distinct entries $a_{j_1,k_1}$ and $a_{j_2,k_2}$,   we have $a_{j_1,k_1}=a_{j_2,k_2}=s$, an ordinary symbol  only if
  \begin{enumerate}
     \item [a.] $j_1\ne j_2$, $k_1\ne k_2$, i.e., they lie in distinct rows and distinct columns; and
     \item [b.] $a_{j_1,k_2}=a_{j_2,k_1}=*$, i.e., the corresponding $2\times 2$  sub-array formed by rows $j_1,j_2$ and columns $k_1,k_2$ must be of the following form
  \begin{IEEEeqnarray}{c}
    \left[\begin{array}{cc}
      s & *\\
      * & s
    \end{array}\right]~\textrm{or}~
    \left[\begin{array}{cc}
      * & s\\
      s & *
    \end{array}\right].
  \end{IEEEeqnarray}
   \end{enumerate}
\end{enumerate}
We refer to the parameter $F$ as the subpacketization level. Specially, if each ordinary symbol $s\in[S]$  occurs exactly
$g$ times, $\bm A$ is called a $g$-$(K,F,Z,S)$ PDA, or $g$-PDA for short.}
\end{definition}
Any PDA can be transformed into a caching scheme having the following performance \cite{YanPDA2017}:
\begin{remark}
	 A $(K,F,Z,S)$ PDA corresponds to a  caching scheme for the shared error-free link setup with $K$ users that is of subpacketization  level $F$, requires cache size    $M=\frac{Z}{F} N$, and delivery rate $R=\frac{S}{F}$.
\end{remark}


Two low-subpacketization schemes were proposed in \cite{YanPDA2017}:

\begin{lemma}[PDA for $\frac{N}{M}\in \mathbb{N}^+$, \cite{YanPDA2017}]\label{lem:1} \emph{For any $q,m\in\mathbb{N}^+, q\geq2$, there exists a $(m+1)$-$(q(m+1),q^m,q^{m-1},q^{m+1}-q^{m})$ PDA, with rate $R=\frac{N}{M}-1$ and subpacketization level $F=(\frac{N}{M})^{\frac{KM}{N}-1}$.}
\end{lemma}

\begin{lemma}[PDA for $\frac{N}{N-M} \in \mathbb{N}^+$, \cite{YanPDA2017}]\label{lem:2}\emph{For any $q,m\in\mathbb{N}^+,q\geq 2$, there exists a $(q-1)(m+1)$-$(q(m+1),(q-1)q^m,(q-1)^2q^{m-1},q^m)$ PDA, with rate $R=\frac{N}{M}-1$ and subpacketization level $F=\frac{M}{N-M}\cdot\big(\frac{N}{N-M}\big)^{K(1-\frac{M}{N})-1}$.}
\end{lemma}

\section{C-PDAs for Combination Networks}\label{sec:main}

A PDA is especially useful for a combination network, if for any coded packet, all the intended users are  connected to the same relay.  This  allows the server to send each coded packet only to this single relay. The following definition ensures the desired property.

\begin{definition}\label{Definition:combinationPDA}\emph{Let $h,r \in \mathbb{N}^+$ with $r \leq
  h$, and $K= {h \choose r}$. A $\left(K,F,Z,S\right)$ PDA  is called
  \emph{$(h,r)$-combinational}, for short \emph{C-PDA},  if its columns can be labeled by
  the sets in $\mathbf{T}$
in a way that for any  ordinary symbol $s\in[S]$, the  labels  of all columns containing symbol $s$ have nonempty intersection.}
\end{definition}
The following example presents  a $(6,6,2,12)$ C-PDA for $h=4$ and $r=2$, and explains how this C-PDA leads to a caching scheme for  the $(4,2)$-combination network  in Fig.~\ref{fig:network}.
\begin{example}\label{exam:CPDA} 
Let $h=4$ and $r=2$. The following table presents a C-PDA combined with a  labeling of the
columns that satisfies the condition in Definition~\ref{Definition:combinationPDA}.
	\begin{table}[!htp]\centering\caption{A C-PDA for the setting  in Fig.~\ref{fig:network}.}\label{eqn:example}
		\small{\begin{tabular}{c|c|c|c|c|c}  \bottomrule
				 $\{1,2\} $& $\{3,4\}$ & $\{1,3\}$ & $\{2,4\}$ & $\{1,4\}$ &
				  $\{2,3\}$\\
				\hline
				 $*$ &$*$& $1$& $4$& $2$& $5$\\
				$1$ & $7 $&$*$ & $* $& $3$ & $6$\\
				$2$ & $8$ & $3$ & $6 $& $*$ & $*$\\
				$*$& $*$ &$ 7$& $10 $& $11$& $8$\\
			 $4$ & $10$ & $*$& $*$& $12$ &$ 9$\\
				 $5$& $11$ & $9$ &$12$ & $*$ & $*$\\\hline
				\toprule
			\end{tabular}}
		\end{table}
		
		The above C-PDA implies the following caching scheme for the $(h=4,r=2)$ combination network in Fig.~\ref{fig:network}. \\
		\underline{1. Placement phase:} Each file  is split into $6$ packets (i.e., the number of rows of the C-PDA), i.e., $W_n=\{W_{n,i}\colon ~i \in [6], ~ n\in [N]\}$. Place the following  cache contents at the users:
		\begin{IEEEeqnarray}{rCcCl}
		Z_{\{1,2\}}&=&Z_{\{3,4\}}&=&\{W_{n,1}, W_{n,4} \colon \ n\in[N]\}\notag\\
		Z_{\{1,3\}}&=&Z_{\{2,4\}}&=&\{W_{n,2},W_{n,5} \colon \ n\in[N]\}\notag\\
		Z_{\{1,4\}}&=&Z_{\{2,3\}}&=&\{W_{n,3},W_{n,6} \colon \ n\in[N]\}\notag
		\end{IEEEeqnarray}
		\underline{2. Delivery phase:} Table \ref{TableII} shows the signals   $X_1,\ldots, X_4$ the server sends to the four relays when users $U_{\{1,2\}},$ $U_{\{3,4\}},$ $U_{\{1,3\}},$ $U_{\{2,4\}},$ $U_{\{1,4\}},$ $U_{\{2,3\}}$ request files  $W_{1},$ $W_{2},$ $W_{3},$ $W_{4},$ $W_{5},$ $W_{6}$, respectively. Each of the coded signals consists of $B/6$ bits, and thus the required rate is $R=1/2$. 
		
		\begin{table}[!htp]\centering\caption{Delivered signals  in Example \ref{exam:CPDA}.}\label{TableII}
			\small{\begin{tabular}{c|c|c|c}  \bottomrule
					Signal & Symbol $s$ & Coded Signal & Intended Users\\\hline
					&$1$&$W_{1,2}\oplus W_{3,1}$& $U_{\{1,2\}},~U_{\{1,3\}}$\\
					$X_1$&$2$&$W_{1,3}\oplus W_{5,1}$&$U_{\{1,2\}},~U_{\{1,4\}}$\\
					&$3$&$W_{3,3}\oplus W_{5,2}$& $U_{\{1,3\}},~U_{\{1,4\}}$\\\hline
					&$4$&$W_{1,5}\oplus W_{4,1}$& $U_{\{1,2\}},~U_{\{2,4\}}$\\
					$X_2$&$5$&$W_{1,6}\oplus W_{6,1}$&$U_{\{1,2\}},~U_{\{2,3\}}$\\
					&$6$&$W_{4,3}\oplus W_{6,2}$& $U_{\{2,4\}},~U_{\{2,3\}}$\\\hline
					&$7$&$W_{2,2}\oplus W_{3,4}$& $U_{\{3,4\}},~U_{\{1,3\}}$\\
					$X_3$&$8$&$W_{2,3}\oplus W_{6,4}$&$U_{\{3,4\}},~U_{\{2,3\}}$\\
					&$9$&$W_{3,6}\oplus W_{6,5}$&$U_{\{1,3\}},~U_{\{2,3\}}$\\\hline
					&$10$&$W_{2,5}\oplus W_{4,4}$&$U_{\{3,4\}},~U_{\{2,4\}}$\\
					$X_4$&11&$W_{2,6}\oplus W_{5,4}$&$U_{\{3,4\}},~U_{\{1,4\}}$\\
					&$12$&$W_{4,6}\oplus W_{5,5}$&$U_{\{2,4\}},~U_{\{1,4\}}$\\
					\toprule
				\end{tabular}}
			\end{table}
			
				Table~\ref{TableII} also indicates the users that are actually interested by each coded signal. In the problem definition, we assumed that each relay forwards its entire received signal to all its connected users. From Table~\ref{TableII}, it is obvious that it would suffice to forward only a subset of the bits to each user.

		\end{example}

We now  present a general way to associate  a
$\left(K,F,Z,S\right)$ C-PDA to a caching scheme for a $(h,r)$-combination network where $h,r$ are positive integers with $r\leq h$.

\textbf{Placement phase:}
 Label the columns of the C-PDA with the  set $\mathbf{T}$ so that  the condition in Definition~\ref{Definition:combinationPDA} is satisfied. Placement is the same as for standard PDAs. That means, split each file $W_{d}$ into $F$ subpackets $(W_{d,1},\ldots, W_{d,F})$ each consisting of $B/F$ bits. Place subfiles $\{W_{n,i}\}_{n=1}^N$ into the cache memory of user $T$, if the C-PDA has entry $``*"$ in row $i$ and the column corresponding to label $T$.   This placement strategy requires a cache  size of  $M=N\cdot \frac{Z}{F}$.

 \textbf{Delivery phase:} The server first creates the coded signals pertaining to each ordinary symbol $s\in[S]$ in the same way as for standard PDAs. It then delivers the coded signal created for each ordinary symbol $s\in[S]$  to one of the relays whose index is contained in the labels of  all columns containing $s$. The \emph{average}  rate  required on the $h$ server-to-relay links is $R_{\textnormal{avg}}= \frac{S}{F h}$.

When in the described scheme  the server sends the same number of bits to each relay, then the following theorem follows immediately from the above description. In fact, in this case subpacketization level $F$ is sufficient. Otherwise, the rate on each server-to-relay link has to be made equal  by first splitting each file into $h$ subfiles and then applying a caching scheme with the same C-PDA but a different  shifted version of the column labels to each of the subfiles.

\begin{theorem}\label{thm:CPDA}\emph{
	Given a $\left(K,F,Z,S\right)$ C-PDA. For any  $(h,r)$ combination network with $K={h \choose r}$, it holds that  $
	R^\star\left( 	M= \frac{N \cdot Z}{F}\right) \leq  \frac{S}{F h}
	$.  This upper bound is achieved by a scheme of subpacketization level not exceeding $hF$.}
 \end{theorem}

 \section{Transforming PDAs into Larger C-PDAs}\label{sec:PDAconstruction}

We present a way of constructing  C-PDAs  for  resolvable $(h,r)$-combination networks (i.e., when $r|h$) from any smaller PDA  that has $\tilde{K}={h-1\choose r-1}$ columns. We start with an example.

\begin{example} Reconsider Example \ref{exam:CPDA}, where $h=4$ and $r=2$, and
	notice that for this resolvable network  (see Definition~\ref{def:resolvable}), a possible partition of $\mathbf{T}$  is  $\mathcal{P}_1=\{\{1,2\},\{3,4\}\}, \mathcal{P}_2=\{\{1,3\},\{2,4\}\}$ and $\mathcal{P}_3=\{\{1,4\},\{2,3\}\}$.
	Consider now the
	$(3,3,1,3)$  PDA of the Maddah-Ali \& Niesen scheme with $\tilde{K}=3$ users:
\begin{IEEEeqnarray}{rCl}
\bm A&=&\left[\begin{array}{ccc}
        * & 1 & 2 \\
        1 & * & 3 \\
        2 & 3 & *
      \end{array}
\right].\notag
\end{IEEEeqnarray}
One can verify that the C-PDA in Table \ref{eqn:example} is obtained from above PDA $\bm A$ by
replicating each column of $\bm A$  first horizontally and  then  each column of the resulting array also vertically,   and by then  replacing the  3 replicas of each ordinary symbol with 3 new (unused)  symbols. The column labels are obtained by labeling  the first two columns of $\bm A$ with the two elements of  $\mathcal{P}_1$, the following two columns  with the  elements of $\mathcal{P}_2$, and the last two columns with the  elements of $\mathcal{P}_3$. 


\end{example}

We now present the general transformation method. We use the following notations. For a given user $T$, let $\delta(T)$ indicate the parallel class that $T$ belongs to, i.e., $\delta(T)=j$ iff $T\in \mathcal{P}_j$.
 Let $T[i]$ be the $i$-th smallest  element of $T$.  For example, if $T=\{2,4\}$, then $T[1]=2, T[2]=4$. Likewise, denote the inverse map by $T^{-1}$, i.e., $T[i]=j$ iff $T^{-1}[j]=i$. 
   \begin{transformation} \label{construction:CPDA}
Given a $(\tilde{K},\tilde{F},\tilde{Z},\tilde{S})$  PDA $\bm{\tilde{C}}=[\tilde{c}_{j,k}]$. Let the following  $(\tilde{F} r)$-by-$(\tilde{K} \frac{h}{r})$  array $\bm C$ be the outcome  applied to PDA $\bm{\tilde{C}}$ for  parameters $(h,r)$: 
\begin{IEEEeqnarray}{c}
\bm C=\left[\begin{array}{cccc}
              \bm{c}_{1,T_1} & \bm{c}_{1,T_2} &\cdots  & \bm{c}_{1,T_K} \\
              \bm{c}_{2,T_1} &\bm{c}_{1,T_2}  & \cdots & \bm{c}_{2,T_K} \\
              \vdots & \vdots & \ddots &\vdots  \\
              \bm{c}_{r,T_1} &\bm{c}_{r,T_2}  &\cdots  &\bm{c}_{r,T_K}
            \end{array}
\right],\notag
\end{IEEEeqnarray}
where $T_{1}, \ldots, T_K$ are the elements of the user set $\mathbf{T}$ in \eqref{eq:Tset}, and $\bm{c}_{i,T_{k}}=[c_{i,j,T_k}]_{j=1}^{\tilde{F}}$ is a single-column array of length $\tilde{F}$, with $j$-th entry
\begin{IEEEeqnarray}{c}
c_{i,j,T_k}=\left\{\begin{array}{ll}
                    *,&\mbox{if}~\tilde{c}_{j,\delta(T_k)}=*,  \\
                   \tilde{c}_{j,\delta(T_k)}+(T_k^{-1}[i]-1)\tilde{S},
                   &\mbox{if}~\tilde{c}_{j,\delta(T_k)}\neq *.
                 \end{array}
\right.\notag
\end{IEEEeqnarray}
   \end{transformation}

%
%

\begin{theorem} \label{thm:comPDA}\emph{Let $h,r$ be positive integers so that $r|h$, and $\tilde{K}={h-1 \choose r-1}$. Applying Transformation~1 with parameters $(h,r)$ to
a	 $(\tilde{K},\tilde{F},\tilde{Z},\tilde{S})$ PDA 
yields a $(K,F,Z,S)$ C-PDA, where 
\begin{IEEEeqnarray}{c}
K={h\choose r}, \quad F=r\tilde{F}, \quad Z=r\tilde{Z}, \quad \textnormal{and} \quad S=h\tilde{S}.\notag
	 \end{IEEEeqnarray}
With the resulting C-PDA,   subpacketization level $F=r\tilde{F}$ is sufficient to achieve the rate $R=\frac{S}{Fh}$.}
\end{theorem}
\begin{IEEEproof}Array $\bm{C}$ satisfies C$1$, C$2$, and C$3$ and is thus a PDA. It also satisfies the condition in Definition~\ref{Definition:combinationPDA}, because $c_{i,j,T_k}= c_{i', j', T_{k'}}=s\in[S]$ implies that $T_{k}^{-1}[i]= T_{k'}^{-1}[i']$, and thus the labels  of all columns containing a given symbol $s$ must have non-empty intersection. The statement on rate and subpacketization follows by Theorem~\ref{thm:CPDA} and the discussion before it. 
 \end{IEEEproof}

The coding scheme for resolvable combination networks in \cite{LiTangISIT2016} can be represented in form of a C-PDA, and this C-PDA can be  obtained by applying Transformation~\ref{construction:CPDA} to the PDA of the Maddah-Ali \& Niesen scheme. Theorem~\ref{thm:comPDA} thus allows to recover the following result from \cite{LiTangISIT2016}.
\begin{corollary}\label{corollary1}\emph{For a $(h,r)$-combination network where $r|h$, when $M\in\{0,\frac{Nh}{Kr},\frac{2Nh}{Kr},\cdots,N\}$,  there exists a caching scheme that requires rate
$R_{\textnormal{TR}}\triangleq \frac{K(1-M/N)}{h\left(1+{KMr}/{(Nh)}\right)}$ and has subpacketization level $F_{\textnormal{TR}}\triangleq r{{Kr}/{h}\choose {KMr}/{(Nh)} }$.}
\end{corollary}

We  apply   Transformation \ref{construction:CPDA} to the reduced versions (so as to have the right number of columns) of the low-subpacketization PDAs in Lemmas \ref{lem:1} and \ref{lem:2}. This yields the first  low-subpacketization C-PDAs and caching schemes for resolvable combination networks. 


\begin{theorem}[C-PDA construction from Lemma \ref{lem:1}]\label{thm:new1} \emph{For any
  $(h,r)$-combination network with $r|h$ and cache sizes
$M \in \{\frac{1}{q}\cdot N:q\in\mathbb{N}^+,q\geq 2\}$, the following upper bound is achieved by a scheme with subpacketization level $F_{\textnormal{LSub1}}\triangleq r\left(\frac{N}{M}\right)^{\lceil\frac{KMr}{Nh}\rceil-1}$:
\begin{IEEEeqnarray}{c}
R^\star(M) \leq R_{\textnormal{LSub1}}\triangleq \frac{1}{r}\cdot\left(\frac{N}{M}-1\right).\notag
	\end{IEEEeqnarray}
	(Here, subscript ``LSub" stands for ``low-subpacketization".)}
\end{theorem}
\begin{IEEEproof} By Lemma \ref{lem:1}, there exists a PDA with $\lceil \frac{\tilde{K}}{q}\rceil q$ columns. Delete any  $\lceil\frac{\tilde{K}}{q}\rceil q-\tilde{K}$ of the columns. Since  each ordinary symbol occurs in $\lceil \frac{\tilde{K}}{q}\rceil$  distinct columns,  some ordinary symbols can be completely deleted whenever $\lceil\frac{\tilde{K}}{q}\rceil q-\tilde{K}\geq \lceil \frac{\tilde{K}}{q}\rceil$. In this case, the reduced PDA has rate smaller than $\frac{N}{M}-1$.
The theorem is concluded by Theorems~\ref{thm:CPDA} and \ref{thm:comPDA}.
\end{IEEEproof}

\begin{theorem}[C-PDA construction from Lemma \ref{lem:2}] \label{thm:new2}\emph{For any
  $(h,r)$-combination network  with $r|h$ and cache sizes $M\in\{\frac{q-1}{q}\cdot N:q\in\mathbb{N}^+,q\geq 2\}$, the following upper bound is achieved by a scheme with subpacketization level $F_{\textnormal{LSub}2}\triangleq\frac{rM}{N-M}\cdot(\frac{N}{N-M})^{\lceil\frac{Kr}{h}\left(1-\frac{M}{N}\right)\rceil-1}$:
 \begin{IEEEeqnarray}{c}
R^\star(M) \leq R_{\textnormal{LSub}2}\triangleq\frac{1}{r}\cdot\left(\frac{N}{M}-1\right).\notag
	\end{IEEEeqnarray}}
\end{theorem}
\begin{IEEEproof}
Similarly to the proof of Theorem~\ref{thm:new1}, except that  deleting $\lceil\frac{\tilde{K}}{q}\rceil q-\tilde{K}$ columns does not  delete any of the ordinary symbols, as each of them occurs $\lceil \frac{\tilde{K}}{q}\rceil(q-1)$ times. 
\end{IEEEproof}


For fair comparison, we compare the new schemes with the scheme  in
\cite{LiTangISIT2016} (Corollary \ref{corollary1}) when $K \leq N$ for the same memory size.  We start with a comparison of the required rates. If $M=\frac{N}{q}$ for some  integer $q\geq2$, then
$
\frac{KMr}{KMr+Nh}\leq\frac{R_{\textnormal{TR}}}{R_{\textnormal{LSub1}}} \leq1.
$
Similarly, if
 $M=\frac{(q-1)N}{q}$ for some  integer $q\geq2$, then
$
\frac{KMr}{KMr+Nh}\leq\frac{R_{\textnormal{TR}}}{R_{\textnormal{LSub2}}} \leq1.
$
As a consequence, if  $M=\frac{N}{q}$ or $M=\frac{(q-1)N}{q}$ for some  integer $q\geq2$, then
\begin{IEEEeqnarray}{c}
\varliminf_{K \to \infty} \frac{R_{\textnormal{TR}}}{R_{\textnormal{LSub1}}} =1 \qquad \textnormal{or} \qquad \varliminf_{K \to \infty} \frac{R_{\textnormal{TR}}}{R_{\textnormal{LSub2}}}  =1.\notag
\end{IEEEeqnarray}

On the other hand, for large values of $K \gg 1$,
by  Corollary~\ref{corollary1} and \cite[Lemma~4]{YanPDA2017}, the subpacketization levels of the schemes satisfy
\begin{IEEEeqnarray}{c}
F_{{\textnormal{TR}}}\sim \sqrt{\frac{N^2hr}{2\pi KM(N-M)}}\cdot e^{\frac{Kr}{h}\left(\frac{M}{N}\ln \frac{N}{M}+(1-\frac{M}{N})\ln \frac{N}{N-M}\right)},\notag
\end{IEEEeqnarray}
and
\begin{IEEEeqnarray}{rCl}
F_{\textnormal{LSub1}}&\leq&
re^{\frac{Kr}{h}\cdot\frac{M}{N}\ln\frac{N}{M}},\notag\\
~F_{\textnormal{LSub2}}&\leq& \frac{rM}{N-M} e^{\frac{Kr}{h}\cdot(1-\frac{M}{N})\ln \frac{N}{N-M}}.\notag
\end{IEEEeqnarray}


As a consequence, if  $M=\frac{N}{q}$ or $M=\frac{(q-1)N}{q}$  for some  integer $q\geq2$, then
\begin{IEEEeqnarray}{c}
\varliminf_{K \to \infty} \frac{F_{\textnormal{TR}}}{F_{\textnormal{LSub1}}} =\infty \qquad \textnormal{or} \qquad \varliminf_{K \to \infty} \frac{F_{\textnormal{TR}}}{F_{\textnormal{LSub2}}}  =\infty.\notag
\end{IEEEeqnarray}

\section{Achieving the Cutset Bound with Low Subpacketization Level}\label{sec:cut_scheme}

Throughout this section, $r,h$ denote positive integers with $r\leq h$. But $r$ does not necessarily divide $h$.

	Let $S_1,\ldots, S_{ h \choose r-1}$ denote all the subsets of $[h]$ of size $r-1$.
Define  $\bm B$ as the   ${h\choose r-1}$-by-${h\choose r}$ dimensional  array  
with  element $b_{j,T}$  in row $j\in \{1,\ldots {h \choose r-1}\}$ and column $T\in\mathbf{T}$, where
\begin{IEEEeqnarray}{c}
b_{j,T}=\left\{\begin{array}{ll}
                  *,&\mbox{if}~S_j \not\subset T,  \\
                 T\backslash S_j, &\mbox{if}~S_j\subset T.
               \end{array}
\right.\label{eq:B}
\end{IEEEeqnarray}
Notice that the set of arrays $\mathbf{B}$ forms a subset of the  PDAs in \cite{YanBipatite2017}. They  can be proved to be C-PDAs.

\begin{example}\label{example:boundconstruction} For  $h=4$ and $r=2$, the C-PDA $\bm B$ is:
	\begin{center}
        \small{\begin{tabular}{c|c|c|c|c|c}  \bottomrule
   $\{1,2\}$ &$\{1,3\}$&$\{1,4\}$&$\{2,3\}$&$\{2,4\}$&$\{3,4\}$\\\hline
 $2$ & $3$ & $4$ & $*$ & $*$ & $*$ \\
 $1$ & $*$ & $*$ & $3$ & $4$ & $*$ \\
 $*$ & $1$ & $*$ & $2$ & $*$ &$4$  \\
 $*$ & $*$ & $1$ & $*$ & $2$ &$3$\\\hline
   \toprule
        \end{tabular}}
\end{center}
\end{example}


The caching scheme corresponding to the C-PDA~$\mathbf{B}$, allows to determine the optimal rate $R^\star(M)$ for large cache sizes $M$. 
\begin{theorem}\label{thm:boundscheme} \emph{For an $(h,r)$-combination network:
	\begin{IEEEeqnarray}{c}
	R^\star(M)= \frac{1}{r}\left(1-\frac{M}{N}\right), \quad M\in \bigg[N\frac{K-h+r-1}{K},\ N\bigg].\notag
	\end{IEEEeqnarray}
	This  can be achieved with subpacketization level $F= {h \choose r-1}$ when $M=N\frac{K-h+r-1}{K}$.}
\end{theorem}
\begin{IEEEproof} The converse follows from the cutset lower bound in 	\cite{MJi2015}. For $M=N\left(1-\frac{h-r+1}{K}\right)$, the upper bound follows by Theorem~\ref{thm:CPDA} and the caching scheme corresponding to
the C-PDA $\mathbf{B}$ in \eqref{eq:B}.
For $M > N\left(1-\frac{h-r+1}{K}\right)$, the upper bound follows by time/memory sharing arguments.  
\end{IEEEproof}



The optimal rate $R^\star(M)$ is in general not achieved by the uncoded placement scheme in  \cite{LiTangISIT2016} (see Corollary \ref{corollary1}). In fact, at the point  $M=N\cdot\left(1-\frac{h-r+1}{K}\right)$,   the scheme in \cite{LiTangISIT2016} requires rate
$R_{\textnormal{TR}}=\frac{1}{r}\left(1-\frac{M}{N}\right)\cdot \frac{Kr}{Kr-(r-1)(h-r)}$, which is strictly
larger than $R^\star(M)$ whenever $r\geq2$. Moreover, it has a subpacketization level
$r{{h-1\choose r-1}\choose r-1}$, which is significantly higher than the one in Theorem \ref{thm:boundscheme}.

%
\section{Conclusion}\label{sec:conclusion}
We introduced the C-PDAs (a subclass of PDAs) to characterize caching schemes with uncoded
placement for combination networks. We also proposed a method to transform certain PDAs to
C-PDAs for resolvable networks. This allowed us to obtain the first low-subpacketization
schemes for resolvable combination networks with a  rate that is close to the rate of the uncoded placement schemes in \cite{LiTangISIT2016}. We also proposed  C-PDAs for general
combination  networks. These C-PDAs have low subpacketization level and achieve the cut-set
lower bound when the cache memories are sufficiently large. 
\section*{Acknowledgment}
The work  of Q. Yan and M. Wigger has been supported by the ERC  Grant \emph{CTO Com}.


\begin{thebibliography}{1}

%

\bibitem{Maddah2014fundamental}M. A. Maddah-Ali and U. Niesen, ``Fundamental limits of caching,"
\emph{IEEE Trans. Inf. Theory,} vol. 60, no. 5, pp. 2856--2867, May 2014

%
%
%

\bibitem{Gaussian2016} S. S. Bidokhti, M. Wigger, and A. Yener, ``Gaussian broadcast channels with receiver cache alignment,"  in Proc. of \emph{ICC}, 2017, May 2017, Paris, France.

\bibitem{Ngo2017}K. H. Ngo, S. Yang, and M. Kobayashi, ``Scalable content delivery with
  coded caching in multi-antenna fading channels," \emph{IEEE Trans. Wireless Commun.}, vol.
  17, no. 1, pp. 548--562, Jan. 2018.
\bibitem{Caire2017} S. P. Shariatpanahi, G. Caire, and B. H. Khalaj, ``Multi-antenna coded caching," in Proc. \emph{ISIT}, 2017, pp. 2113--2117, Jul. 2017, Aachen, Germany.
\bibitem{Elia2017}J. Zhang, and P. Elia, ``Feedback-aided coded caching for the MISO BC with small caches," in Proc. of \emph{ICC}, 2017, May 2017, Paris, France.
\bibitem{MJi2015}M. Ji, M. F. Wong, A. M. Tulino, J. Llorca, G. Caire, M. Effros, and M.  Langberg, ``On the fundamental limits of caching in combination networks,"  in Proc. \emph{SPAWC}, 2015, pp. 695--699, Jun. 2015,  Stockholm, Sweden.


\bibitem{LiTangISIT2016}L. Tang and A. Ramamoorthy, ``Coded caching for networks with resolvability property," in Proc. \emph{ISIT}, 2016,  pp. 420--424, Jul. 2016, Barcelona, Spain.


\bibitem{Wan2017Bound}K. Wan, M. Ji, P. Piantanida, and D. Tuninetti, ``Novel outer bounds and inner bounds with uncoded cache placement for combination networks with end-user-caches", arXiv:1701.06884v5.
\bibitem{Wan2018new} K. Wan, D. Tuninetti, P. Piantanida, and M. Ji,  ``On combination networks with cache-aided relays and users," arXiv:1803.06123.
\bibitem{Zewail2017conf} A. A. Zewail and A. Yener, ``Coded caching for combination networks with
cache-aided relays," in Proc. \emph{ISIT}, 2017, pp. 2433--2437, Jun. 2017,  Aachen, Germany.
\bibitem{combination2017} A. A. Zewail and A. Yener, ``Combination networks with or without secrecy constraints: the impact of caching relays," arXiv:1712.04930.
\bibitem{Wan2018asymmetric} K. Wan, D. Tuninetti, M. Ji, and P. Piantanida, ``A novel asymmetric coded placement in combination networks with end-user caches," 	arXiv:1802.10481.







\bibitem{Shanmugam2016}K. Shanmugam, M. Ji, A. M. Tulino, J. Llorca, and A. G. Dimakis,
``Finite-length analysis of caching-aided coded multicasting," \emph{IEEE
Trans. Inf. Theory,} vol. 62, no. 10, pp. 5524--5537, Oct. 2016.


\bibitem{YanPDA2017}Q. Yan, M. Cheng, X. Tang, and Q. Chen, ``Placement delivery array design for centralized coded caching scheme," \emph{IEEE Trans. Inf. Theory,} vol. 63, no. 9, pp. 5821--5833, Sep. 2017.


\bibitem{YanBipatite2017}Q. Yan, X. Tang, Q. Chen, and M. Cheng, ``Placement delivery array design through strong edge coloring of bipartite graphs", \emph{IEEE Commun. Lett.}, vol. 22, no. 2, pp. 236--239, Feb. 2018.

\bibitem{LiTang}L. Tang, and A. Ramamoorthy, ``Low subpacketization schemes for coded caching," in Proc. \emph{ISIT}, 2017, pp. 2790--2794, Jul. 2017, Aachen, Germany.
\bibitem{Ge2016}C. Shangguan, Y. Zhang, and G. Ge, ``Centralized coded caching schemes:
A hypergraph theoretical approach," arXiv:1608.03989.

\bibitem{linear2017}K. Shanmuguam, A. M. Tulino, and A. G. Dimakis, ``Coded caching with linear subpacketization is possible using Ruzsa-Szem$\acute{\mbox{e}}$redi graphs," in Proc. \emph{ISIT}, 2017, pp. 1237--1241, Jul. 2017, Aachen, Germany.


\bibitem{C_PDA2018arXiv}Q. Yan, M. Wigger, and S. Yang, ``Placement delivery array design for combination networks with edge caching," arXiv:1801.03048.












\end{thebibliography}
\end{document}